\documentclass[letterpaper, 10 pt, conference]{ieeeconf}  

\IEEEoverridecommandlockouts                              
\usepackage{bm}
\usepackage{srcltx}
\usepackage{cite}    
\usepackage{verbatim}
\usepackage{float}
\usepackage{color}
\usepackage{acronym}
\usepackage{caption}
\usepackage{subcaption}
\usepackage{mathtools}
\usepackage{cite}

\usepackage[lined,ruled,vlined]{algorithm2e}

\usepackage{amssymb,latexsym,amsfonts,amsmath,amscd} 
\usepackage{paralist}

\usepackage{tikz}
\usetikzlibrary{shapes,snakes}

\usetikzlibrary{backgrounds}	

\usepackage{placeins}
\usetikzlibrary{arrows,fit,shapes,automata}
\usetikzlibrary{positioning,fit,calc,shapes}
\usetikzlibrary{decorations.fractals}
\usetikzlibrary{decorations.markings}

 \usepackage{acronym}

\providecommand{\card}[1]{\mathsf{card}(#1)}

\providecommand{\norm}[1]{\lVert #1 \rVert}

\newcommand{\sink}{\mathsf{sink}}
\newcommand{\calA}{\mathcal{A}}

\newcommand{\AEC}{\mathsf{AEC}}

\newcommand{\calD}{\mathcal{D}}

\newcommand{\calM}{\mathcal{M}}

\newcommand{\calAP}{\mathcal{AP}}

\newcommand{\prox}{\mathbf{prox}}
\newcommand{\proj}{\mathbf{proj}}
\newcommand{\avg}{\mathbf{avg}}
\newcommand{\exch}{\mathbf{exch}}
\newcommand{\truev}{\mathsf{true}}

  \newtheorem{definition}{Definition}
 \newtheorem{problem}{Problem}

\newtheorem{lemma}{Lemma}

\newtheorem{remark}{Remark}

\acrodef{dba}[DBA]{deterministic B\"uchi automaton}
\acrodef{molp}[MOLP]{multi-objective linear programming}
\acrodef{admm}[ADMM]{alternating direction method of multipliers}
\acrodef{lp}[LP]{linear programming}
\acrodef{ltl}[LTL]{Linear temporal logic}
\acrodef{mdp}[MDP]{Markov decision process}
\acrodef{mdps}[MDPs]{Markov decision processes}

\acrodef{rl}[RL]{reinforcement learning}
\acrodef{dra}[DRA]{deterministic Rabin automaton}

\newcommand{\supp}{\mathsf{Supp}}

\newcommand{\acc}{\mathsf{Acc}} 
\newcommand{\Inf}{\mathsf{Inf}} 

\newcommand{\calC}{\mathcal{C}}

\newcommand{\act}{A}

\newcommand{\val}{\mathsf{Val}}

\newcommand{\Peri}{\mathsf{Periphery}}

\acrodef{aec}[AEC]{accepting end components}
\usepackage{graphicx}
\graphicspath{{./figure/}}

\begin{document}

\title{Optimal control in Markov decision processes via distributed
  optimization \author{Jie Fu, Shuo Han and Ufuk Topcu} \thanks{J. Fu,
    S. Han and U. Topcu are with the Department of Electrical and
    Systems Engineering, University of Pennsylvania, Philadelphia, PA
    19104, USA {\tt\small jief, hanshuo, utopcu@seas.upenn.edu}.}}
\maketitle
\begin{abstract}
  Optimal control synthesis in stochastic systems with respect to
  quantitative temporal logic constraints can be formulated as linear
  programming problems. However, centralized synthesis algorithms do
  not scale to many practical systems. To tackle this issue, we
  propose a decomposition-based distributed synthesis algorithm. By
  decomposing a large-scale stochastic system modeled as a Markov
  decision process into a collection of interacting sub-systems, the
  original control problem is formulated as a linear programming
  problem with a sparse constraint matrix, which can be solved through
  distributed optimization methods.  Additionally, we propose a
  decomposition algorithm which automatically exploits, if exists, the
  modular structure in a given large-scale system. We illustrate the
  proposed methods through robotic motion planning examples.
\end{abstract}
\section{Introduction}
For many systems, temporal logic formulas are used to describe
desirable system properties such as safety, stability, and liveness
\cite{manna1992temporal}. Given a stochastic system modeled as a
\ac{mdp}, the synthesis problem is to find a policy that achieves
optimal performance under a given quantitative criterion regarding
given temporal logic formulas. For instance, the objective may be to
find a policy that maximizes the probability of satisfying a given
temporal logic formula. In such a problem, we need to keep track of
the evolution of state variables that capture system dynamics as well
as predicate variables that encode properties associated with the
temporal logic constraints
\cite{thiebaux2006decision,baier2008principles}.  As the number of
states grows exponentially in the number of variables, we often
encounter large \ac{mdp}s, for which the synthesis problems are
impractical to solve with centralized methods. 
The
insight for control synthesis of large-scale systems is to exploit the
modular structure in a system so that we can solve the original
problem by solving a set of small subproblems.

In literature, distributed control synthesis methods are proposed in
the pioneering work for \ac{mdp}s with discounted rewards
\cite{kushner1974decomposition,dean1995decomposition}. The authors
formulate a two-stage distributed reinforcement learning method: The
first stage constructs and solves an abstract problem derived from the
original one, and the second stage iteratively computes parameters for
local problems until the collection of local problems' solutions
converge to one that solves the original problem.  Recently, \ac{admm}
is combined with a sub-gradient method into planning for
average-reward problems in large MDPs in
\cite{Krishnamurthy2013}. However, the method in
\cite{Krishnamurthy2013} applies only when some special conditions are
satisfied on the costs and transition kernels. Alternatively,
hierarchical reinforcement learning introduces
\emph{action-aggregation} and \emph{action-hierarchies} to address the
planning problems with large MDPs \cite{barto2003recent}. In
action-aggregation, a micro-action is a local policy for a subset of
states and the global optimal policy maps histories of states into
micro-actions. However, it is not always clear how to define the
action hierarchies and how the choice of hierarchies affects the
optimality in the global policy. Additionally, the aforementioned
methods are in general difficult to implement and cannot handle
temporal logic specifications.

For synthesis problems in \ac{mdp}s with quantitative temporal logic
constraints, centralized methods and tools
\cite{baier2008principles,KNP11} are developed and applied to control
design of stochastic systems and robotic motion planning
\cite{ding2011mdp,wolff2012optimal,Lahijanian2012}. 
Since centralized algorithms are based on either value iteration or
linear programming, they inevitably hit the barrier of scalability and
are not viable for large \ac{mdp}s.

In this paper, we develop a distributed optimization method for large
\ac{mdp}s subject to temporal logic contraints.  We first introduce a
decomposition method for large \ac{mdp}s and prove a property in such
a decomposition that supports the application of the proposed
distributed optimization. For a subclass of \ac{mdp}s whose graph
structures are Planar graphs, we introduce an efficient decomposition
algorithm that exploits the modular structure for the underlying
\ac{mdp} caused by loose coupling between subsets of states and its
constituting components. Then, given a decomposition of the original
system, we employ a distributed optimization method called \emph{block
  splitting algorithm} \cite{Parikh2013} to solve the planning problem
with respect to discounted-reward objectives in large \ac{mdp}s and
average-reward objectives in large ergodic \ac{mdp}s.  Comparing to
two-stage methods in
\cite{dean1995decomposition,Krishnamurthy2013,Daoui2010}, our method
concurrently solves the set of sub-problems and penalizes solutions'
mismatches in one step during each iteration, and is easy to
implement.  Since the distributed control synthesis is independent
from the way how a large \ac{mdp} is decomposed, any decomposition
method can be used. Lastly, we extend the method to solve the
synthesis problems for \ac{mdp}s with two classes of quantative
temporal logic objectives. Through case studies we investigate the
performance and effectiveness of the proposed
method. 





\section{Preliminaries}
\label{sec:prelim}
Let $\Sigma$ be a finite set. Let $\Sigma^\ast, \Sigma^\omega$ be the
set of finite and infinite words over $\Sigma$. $\card{\Sigma}$ is the
cardinality of the set $\Sigma$.  A probability distribution on a
finite set $S$ is a function $D : S \rightarrow [0,1]$ such that
$\sum_{s\in S} D(s)=1$. The support of $D$ is the set
$\supp(D)=\{s\in S\mid D(s) >0\}$. The set of probability
distributions on a finite set $S$ is denoted $\mathcal{D}(S)$.

\vspace{1ex}
\noindent  \textbf{Markov decision process:} A \emph{Markov
  decision process} (MDP) $ M= \langle S, A , u_0, P\rangle $ consists
of a finite set $S$ of states, a finite set $A$ of actions, an initial
distribution $u_0 \in \calD(S)$ of states, and a transition probability
function $P: S \times A \rightarrow \calD(S)$ that for a state
$s\in S$ and an action $a\in A$ gives the probability $P(s,a)(s')$ of
the next state $s'$.  Given an \ac{mdp} $M$ we define the set of
actions \emph{enabled} at state $s$ as
$\act(s)=\{a\in A\mid \exists s' \in S, P(s,a)(s') >0\}$.  The
cardinality of the set $\{(s, a)\mid s \in S, a\in A(s)\}$ is the
number of \emph{state-action pairs} in the \ac{mdp}.

\vspace{1ex}
\noindent \textbf{Policies:} A \emph{path} is an infinite sequence
$s_0s_1\ldots $ of states such that for all $i\ge 0$, there exists
$ a\in A$, $s_{i+1}\in \supp(P(s_i,a))$.  A \emph{policy} is a
function $f: S^\ast \rightarrow \calD(A)$ that, given a finite state
sequence representing the history, chooses a probability distribution
over the set $A$ of actions. Policy $f$ is memoryless if it only
depends on the current state, i.e., $f:S\rightarrow \calD(A)$. Once a
policy $f$ is chosen, an \ac{mdp} $M$ is reduced to a Markov chain,
denoted $M^f$. We denote by $X_i$ and $\theta_i$ the random variables
for the $i$-th state and the $i$-th action in this chain $M^f$. Given
a policy $f$, for a measurable function $\phi$ that maps paths into
reals, let $E_{u_0}^f[\phi]$ (resp. $E_s^f[\phi]$) be the expected
value of $\phi$ when the policy $f$ is used given $u_0$ being the
initial distribution of states (resp. $s$ being the initial state).

\vspace{1ex}
\noindent \textbf{Rewards and objectives:} Given an \ac{mdp} $M$, a
reward function $R: S\times A\rightarrow \mathbb{R}$ and a policy $f$, let $\gamma $ be a discounting factor, the
\emph{discounted-reward value} is defined as $ \val_\gamma^f(u_0)=
\val_\gamma^f(s)\cdot u_0(s)$ where
$\val_\gamma^f(s) = E^f_{s}(\sum_{n=0}^\infty \gamma^n R(X_n,
\theta_n))$;
the \emph{average-reward value} is defined as $\val^f(u_0)=
\val^f(s)\cdot u_0(s)$ where $\val^f (s) = \lim_{n\rightarrow \infty}
\frac{1}{n}E_s^f\left[ \sum_{k=0}^n R(X_k,
  A_k)\right]$. A discounted-reward (resp. an average-reward) problem
is, for a given initial state distribution, to obtain a policy that
maximizes the discounted-reward value (resp.  average-reward value).
For discounted-reward (average-reward) problems, the optimal value can
be attained by memoryless policies \cite{Bertsekas2007}.

A solution to the
discounted-reward problem can be found by solving the \ac{lp} problem:
\begin{subequations}\label{eq:constraintdiscounted}
\begin{align}
&\max_{x \in \mathbb{R}^{m}_{+}} \sum_{s\in S}\sum_{a\in \act(v)} x(s,a)\cdot
R(s,a) \label{eq:constraintdiscounted-obj}\\
&\text{subject to } \\
&  \sum_{a\in \act(s)}  x(s,a) - \gamma \cdot \sum_{s'\in
  S}\sum_{a'\in\act(s')} x(s',a') \cdot P(s',a')(s) \nonumber \\
&= u_0(s), \forall s
\in S, \label{eq:constraintdiscounted-1}
\end{align}
\end{subequations}
where $m$
is the total number of state-action pairs in the \ac{mdp},
$\mathbb{R}^m_+$
is the non-negative orthant of $\mathbb{R}^{m}$,
and variable $x(s,a)$
can be interpreted as the expected discounted time of being in state
$s$
and taking action $a$.
Once the \ac{lp} problem in \eqref{eq:constraintdiscounted} is solved,
the optimal policy is obtained as $f(s,a)
= \frac{x(s,a)}{\sum_{a'\in \act(s)}
  x(s,a')}$ and the objective function's value is the optimal
discounted-reward value under policy $f$
given the initial distribution $u_0$ of states.

In an ergodic \ac{mdp}, the average-reward value is a constant
regardless of the initial state distribution \cite{puterman2009markov}
(see \cite{puterman2009markov} for the definition of
ergodicity). 
We obtain an optimal policy for an average-reward problem by solving the
\ac{lp} problem
\begin{subequations}
\label{eq:averageLP}
\begin{align}
&\max_{x\in \mathbb{R}^m_+}\sum_{s\in S}\sum_{a \in \act(s)} x(s,a)
  \cdot R(s,a) \label{eq:averageLP-obj}\\
&\text{subject to } \nonumber \\
& \sum_{a \in \act(s)} x(s,a) - \sum_{s'\in S} \sum_{a'\in \act(s')}
 x(s',a')  \cdot P(s',a')(s) = 0 ,\nonumber \\
&\forall s \in S,  \label{eq:averageLP-1}\\
& \sum_{s\in S}\sum_{a\in \act(s)} x(s,a)=1 ,\label{eq:averageLP-2}
\end{align}
\end{subequations}
where $x(s,a)$ is understood as the long-run fraction of time that the
system is at state $s$ and the action $a$ is taken.  Once the \ac{lp}
problem in \eqref{eq:averageLP} is solved, the optimal policy is obtained as
$f(s,a) = \frac{x(s,a)}{\sum_{a'\in \act(s)} x(s,a')}$.  The optimal
objective value is the optimal average-reward value and is the same for all
states.

\vspace{1ex}
\noindent  \textbf{Distributed optimization:} As a prelude to the distributed
synthesis method developed in section~\ref{sec:planning}, now we
describe the \emph{alternating direction method of multipliers}
(\ac{admm}) \cite{boyd2011distributed} for the generic convex
constrained minimization problem
$ \min_{z \in \mathbf{C}} g(z)$ where function $g$ is closed proper
convex and set $\mathbf{C}$ is closed nonempty convex. In iteration
$k$ of the \ac{admm} algorithm the following updates are performed:
\begin{subequations}\label{eq:subeqns}
\begin{align}
 z^{k+1/2} & := \prox_{g}(z^k -\tilde z^k), \label{eq:first}\\
z^{k+1} &:= \Pi_{\bf C}(z^{k+1/2} +\tilde z^k), \label{eq:second}\\
\tilde z^{k+1} &:= \tilde z^k + z^{k+1/2} -z^{k+1} \label{eq:dualupdate},
\end{align}
\end{subequations}
where $z^{k+1/2}$ and $\tilde z^k$ are auxiliary variables,
$\Pi_{\bf C}$ is the (Euclidean) projection onto $\bf C$, and
$ \prox_{g}(v)= \arg\min_{x}\left( g(x)+(\rho/2)\norm{x-v}_2^2 \right)
$
is the \emph{proximal operator} \cite{boyd2011distributed} of $g$ with
parameter $\rho >0 $. The algorithm handles separately the objective
function $g$ in \eqref{eq:first} and the constraint set $\bf C$ in
\eqref{eq:second}. In \eqref{eq:dualupdate} the dual update step
coordinates these two steps and results in convergence to a solution
of the original problem. 

\vspace{1ex}
\noindent \textbf{Temporal logic:} \ac{ltl} formulas are defined by:
$ \phi:=p \mid \neg \phi \mid \phi_1\lor \phi_2 \mid \bigcirc \phi
\mid \phi_1\mathcal{U} \phi_2 , $
where $p\in \calAP$ is an atomic proposition, and $\bigcirc$ and
$\mathcal{U}$ are temporal modal operators for ``next'' and
``until''. Additional temporal logic operators are derived from basic
ones: $\lozenge \varphi \coloneqq \truev\; \mathcal{U}\varphi$
(eventually) and
$\square \varphi \coloneqq \neg \lozenge \neg \varphi$.  Given an
\ac{mdp} $M$, let $\calAP$ be a finite set of atomic propositions, and
a function $L: S \rightarrow 2^{\calAP}$ be a labeling function that
assigns a set of atomic propositions $L(s)\subseteq \calAP$ to each
state $s \in S$ that are valid at the state $s$. $L$ can be extended
to paths in the usual way, i.e., $L(s \rho)=L(s)L(\rho)$ for
$s\in S, \rho\in S^\omega$.  A path $\rho=s_0s_1\ldots \in S^\omega$
satisfies a temporal logic formula $\varphi$ if and only if $L(\rho)$
satisfies $\varphi$. In an \ac{mdp} $M$, a policy $f$ induces a
probability distribution over paths in $S^\omega$. The probability of
satisfying an \ac{ltl} formula $\varphi$ is the sum of probabilities
of all paths that satisfy $\varphi$ in the induced Markov chain
$M^f$.

\begin{problem}
\label{prob}
Given an \ac{mdp} $M$ and an \ac{ltl} formula $\varphi$,
synthesize a policy that optimizes a quantitative performance
measure with respect to the formula $\varphi$ in the \ac{mdp} $M$.
\end{problem}
We consider the probability of satisfying a temporal logic formula as
one quantitative performance measure. We also consider the expected
frequency of satisfying certain recurrent properties specified in an
\ac{ltl} formula. 

By formulating a product \ac{mdp} that incorporates the underlying
dynamics of a given \ac{mdp} and its temporal logic specification, it
can be shown that Problem~\ref{prob} with different quantitative
performance measures can be formulated through pre-processing as
special cases of discounted-reward and average-reward problems
\cite{baier2008principles,brazdil2011two}. Thus,
in the following, we first introduce decomposition-based distributed
synthesis methods for large \ac{mdp}s with discounted-reward and
average-reward criteria. Then, we show the extension for solving
\ac{mdp}s with quantitative temporal logic constraints. 
\section{Decomposition of an \ac{mdp}}
\label{sec:decomp}
\subsection{Decomposition and its property}
To exploit the modular structure of a given \ac{mdp}, the initial step
is to decompose the state space into small subsets of states, each of
which can then be related to a small problem.  In this section, we
introduce some terminologies in decomposition of \ac{mdp}s from
\cite{dean1995decomposition}.

Given an \ac{mdp} $M=\langle S, A, u_0,P\rangle$, let $\Pi$ be any
partition of the state set $S$. That is,
$\Pi= \{S_1,\ldots, S_N\}\subseteq 2^S$, $\emptyset \notin \Pi$,
$S_i \cap S_j =\emptyset$ when $i\ne j$ and $\bigcup_{i=1}^N S_i=S$.
A set in $\Pi$ is called a \emph{region}.  The \emph{periphery} of a
region $S_i$ is a set of states \emph{outside} $S_i$, each of which
can be reached with a non-zero probability by taking some action from
a state \emph{in} $S_i$. Formally,
$ \Peri(S_i)=\{s' \in S\setminus S_i \mid \exists (s, a) \in S_i\times
A, P(s,a)(s') >0\}.  $

Let $K_0 = \bigcup_{i=1}^N \Peri(S_i)$.  Given a region $ S_i\in \Pi$,
we call $K_i = S_i\setminus K_0$ the \emph{kernel} of $S_i$. We denote
$m_i$ the number of state-action pairs restricted to $K_i$, for each
$i=0,\ldots, N$. That is, $m_i$ is the cardinality of the set
$\{(s,a)\mid s\in K_i, a \in A(s)\}$. We call the partition
$\{K_i \mid 0\le i \le N\}$ a \emph{decomposition} of $M$.
The following property of a decomposition is exploited in distributed optimization.
\begin{lemma}
\label{lm:kernel}
Given a decomposition $\{K_i, i=0,1,\ldots, N\}$ obtained from
partition $\Pi = \{S_1,\ldots, S_N\}$, for a state $s\in K_i$ where
$i \ne 0$, if there is a state $s'$ and an action $a$ such that
$P(s', a)( s) \ne 0$, then either $s'\in K_0$ or $s'\in K_i$.
\end{lemma}

\begin{proof} 
  Suppose $s' \notin K_0$ and $s'\notin K_i$, then it must be the case
  that $s'\in K_j$ for some $j\ne 0 $ and $j\ne i$. Since from state
  $s' \in S_j$, after taking action $a$, the probability of reaching
  $s\in S_i$ is non-zero, we can conclude that $s\in \Peri(S_j)$,
  which implies $s\in K_0$. The implication contradicts the fact that
  $s\in K_i$ since $K_0\cap K_i=\emptyset$. Hence, either $s'\in K_i$
  or $s'\in K_0$.
\end{proof}
\noindent \textbf{Example:} Consider the \ac{mdp} in Figure~\ref{fig:exmdp},
which is taken from \cite{Baire2004}. The shaded region shows a
partition $\Pi=\{S_1= \{s_4,s_5,s_6\},S_2=\{s_0,s_1,s_2,s_3,s_7\}\}$
of the state space. Then, $\Peri(S_1) =\{s_2,s_7\}$ and
$\Peri(S_2) =\{s_4\}$. We obtain a decomposition of $M$ as
$K_0= \cup_{i=1}^2 \Peri(S_i)= \{s_2,s_4,s_7\}$,
$K_1= S_1\setminus K_0= \{s_5, s_6\}$, and
$K_2=S_2\setminus K_0= \{s_0, s_1, s_3\}$. 
It is observed that state $s_5$ can only be
reached with non-zero probabilities by actions taken from states $s_4$
and $s_6$.
\vspace{-3ex}
 \begin{figure}[ht]
\centering
\includegraphics[width=0.5\textwidth]{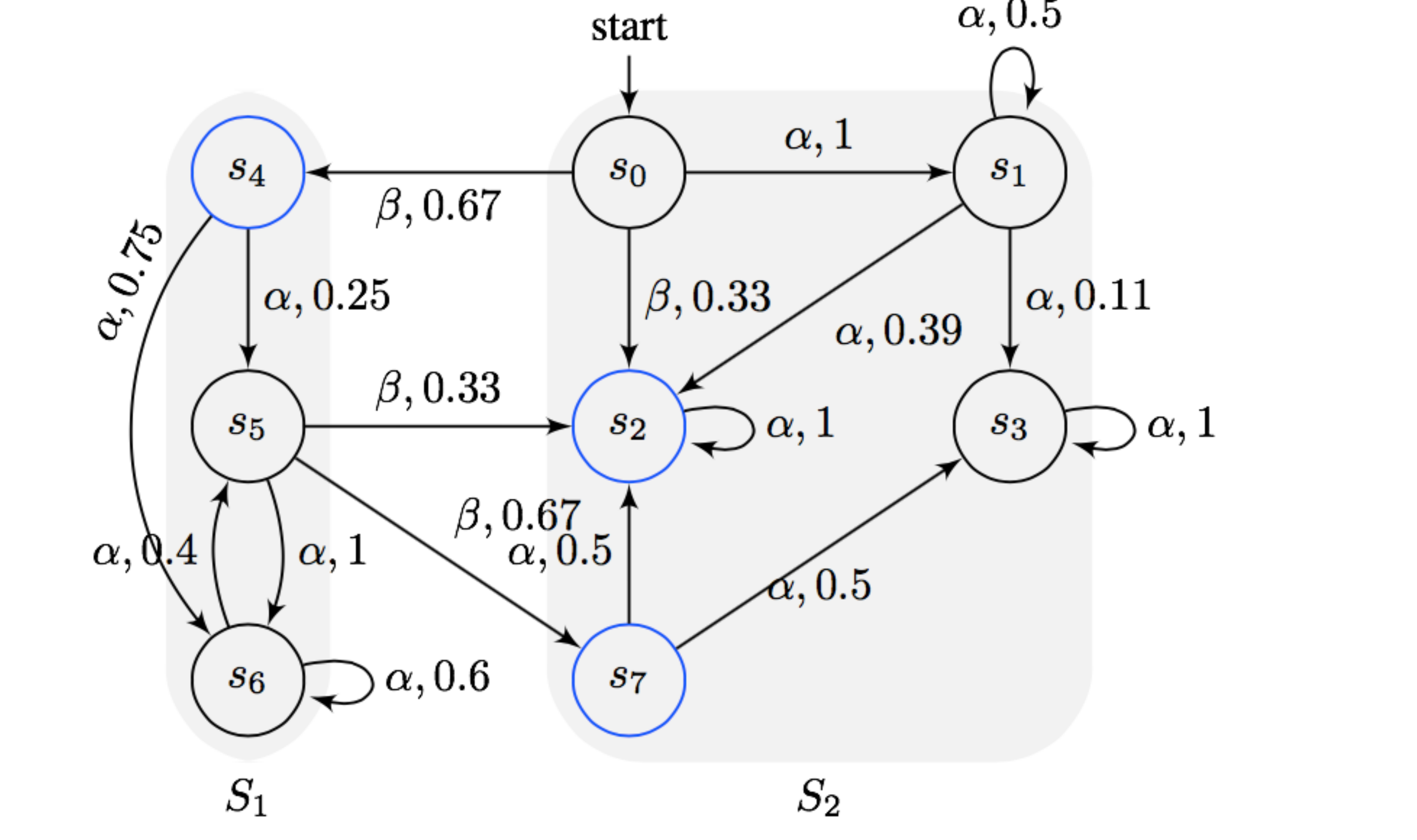}
 \caption{Example of an \ac{mdp} with states $Q=\{s_i, i =0,\ldots,
   7\}$, actions $A=\{\alpha, \beta\}$, and transition
   probability function $P$ as indicated.}
\label{fig:exmdp} 
\end{figure}



\subsection{A decomposition method for a subclass of \ac{mdp}s}
\label{sec:decompmethod}

Various methods are developed to derive a decomposition of an
\ac{mdp}, for example, decompositions based on partitioning the state
space of an \ac{mdp} according to the communicating classes in the
induced graph (defined in the following) of that \ac{mdp} (see a
survey in \cite{Daoui2010}).  For the distributed synthesis method
developed in this paper, it will be shown later in
Section~\ref{sec:planning} that the number of state-action pairs and
the number of states in $K_i$ are the number of variables and the
number of constraints in a sub-problem, respectively.  Thus, we prefer
a decomposition that meets one simple desirable property: For each
$i=0,1,\ldots, N$, the number $m_i$ of state-action pairs in $K_i$ is
small in the sense that the classical linear programming algorithm can
efficiently solve an \ac{mdp} with state-action pairs of this
size. Next, we propose a method that generates decompositions which
meet the aforementioned desirable property for a subclass of
\ac{mdp}s. For an \ac{mdp} in this subclass, its induced graph is
\emph{Planar} \footnote{A graph is planar if it can be drawn in the
  plane in such a way that no two edges meet each other except at a
  vertex to which they are incident.}. It can be shown that \ac{mdp}s
derived from classical gridworld examples, which have many practical
applications in robotic motion planning, are in this
subclass.
We start by relating an \ac{mdp} with a directed graph.
\begin{definition}The \emph{labeled digraph} induced from an \ac{mdp}
  $M =\langle S, A, u_0, P\rangle $ is a tuple
  $G= \langle S, E\rangle$ where $S$ is a set of nodes, and
  $E\subseteq S\times A\times S$ is a set of labeled edges such that
  $(s,a, s')\in E$ if and only if $P(s,a)(s') >0$.
\end{definition}
Let $n = \card{S}$ be the total number of nodes in the graph. A
partition of states in the \ac{mdp} gives rise to a partition of nodes
in the graph. Given a partition $\Pi$ and a region $S_i \in \Pi$, a
node is said to be \emph{contained} in $S_i$ if some edge of the
region is incident to the node. A node contained in more than one
regions is called a \emph{boundary} node. That is, $s \in S_i$ is a
boundary node if and only if there exists $(s,a,s') \in E$ or
$(s',a, s)\in E$ with $s' \notin S_i$. Formally, the boundary nodes of
$S_i$ are $B_i = \mathsf{In}_i \cup \mathsf{Out}_i $ where
$\mathsf{In}_i = \{s\in S_i \mid \exists j\ne i, s'\in S_j ,a\in
A(s'), \text{ and } (s', a, s) \in E\} $
and
$\mathsf{Out}_i= \{s \in S_i \mid \exists s'\in S\setminus S_i, a\in
A(s), \text{ and } (s, a, s') \in E\} $.
We define $B_0=\bigcup_{i=1}^N B_i $. Note that since
$\bigcup_{i=1}^N \mathsf{In}_i = K_0$, $K_0 \subseteq B_0$.
We use the number of boundary nodes as an upper bound on the
size of the set $K_0$ of states.


\begin{definition}\cite{federickson1987fast} 
An \emph{$r$-division
  of an $n$-node graph} is a  partition of nodes into $O(n/r)$
subsets, each of which have $O(r)$ nodes and $O(\sqrt{r})$
boundary nodes. 
\end{definition} 

Reference \cite{federickson1987fast} shows an algorithm that divides a planar graph of $n$ vertices
into an $r$-division in $O(n \log n)$ time.

\begin{lemma}
\label{lma:decomp}Given a partition $\Pi$ of an \ac{mdp} with $n$ states obtained with a
  $r$-division the induced graph, the number of states in $K_0$ is
  upper bounded by $O(n/\sqrt{r})$ and the number of
  states in $K_i$ is upper bounded by $O(r)$.
\end{lemma}
\begin{proof}
  Since each boundary node is contained in at most three regions and
  at least one region by the property of an $r$-division
  \cite{federickson1987fast}, the total number of boundary nodes is
  $O(\sqrt r \cdot \frac{n}{r}) =O(n/\sqrt{r})$. The number of states
  in $K_i$ is upper bounded by the size of $S_i$, which is $O(r)$.
\end{proof}

To obtain a decomposition, the user specifies an approximately upper
bound on the number of variables for all sub-problems. Then, the
algorithm decides whether there is an $r$-division for some $r$ that
gives rise to a decomposition that has the desirable property.

\begin{remark}
  Although the decomposition method proposed here is applicable for a
  subclass of \ac{mdp}s, the distributed synthesis method developed in
  this paper does not constrain the way by which a decomposition is
  obtained. A decomposition may be given or obtained
  straight-forwardly by exploiting the existing modular structure of
  the system. Even if a decomposition does not meet the desirable
  property for the distributed synthesis method, the proposed method
  still applies as long as each subproblem derived from that
  decomposition (see Section~\ref{sec:admm}) can be solved given
  the limitation in memory and computational capacities.
\end{remark}

\section{Distributed synthesis: Discounted-reward and average-reward
  problems}
\label{sec:planning}
In this section, 
we show that under a decomposition, the original \ac{lp} problem for a
discounted-reward or average-reward case can be formulated into one
with a sparse constraint matrix. Then, we employ block-splitting
algorithm based on \ac{admm} in \cite{Parikh2013} for solving the
\ac{lp} problem in a distributed manner.

\subsection{Discounted-reward case}

Given a decomposition $\{K_i \mid i=0,1,\ldots, N\}$ of an \ac{mdp},
let $x_i$ be a vector consisting of variables $x(s,a)$ for all
$s\in K_i$ with all actions enabled from $s$.  Let
$\iota_i: K_i\rightarrow \{1, \ldots, \card{K_i}\}$ be an index
function.  The constraints in \eqref{eq:constraintdiscounted-1} can be
written as: For each $s\in K_0$,
\begin{multline}
\label{eq:k0discounted}\sum_{a\in \act(s)} x(s,a)= u_0(s)+ \\\gamma
\cdot  \sum_{i=0}^N \sum_{s'\in K_i}\sum_{a'\in \act(s')} x(s',a')\cdot
P(s',a')(s),
\end{multline}
and for each $s\in K_i$, $i=1,\ldots, N$, 
\begin{multline}
\label{eq:kidiscounted}
\sum_{a\in \act(s)} x(s,a)= u_0(s)+\\ \gamma \cdot \big(
\sum_{s'\in K_0}\sum_{a'\in \act(s')} x(s',a')\cdot
P(s',a')(s) \\
+ \sum_{s'\in K_i}\sum_{a'\in \act(s')} x(s',a')\cdot P(s',a')(s) \big).
\end{multline}
Recall that, in Lemma~\ref{lm:kernel}, we have proven that each
$s\in K_i$ with $i\ne 0$ can only be reached with non-zero probabilities from states
in $K_0$ and $K_i$. As a result, 
for each state $s$ in $K_i$ with $i\ne 0$ and each action $a \in \act(s)$,
the constraint on variable $x(s,a)$ is only related with variables in
$x_i$ and $x_0$.  Let $x = (x_0, x_1,\ldots, x_N)$. We denote the
number of variables in $x_i$ by $m_i$ and the number of states in the
set $K_i$ by $n_i$.  Let $m=\sum_{i=0}^N m_i$.  The \ac{lp} problem in
\eqref{eq:constraintdiscounted} is then
\begin{equation}
\label{eq:transformdiscountedLP}
  \min _{x \in \mathbb{R}^m_+}\sum_{j=0}^N c_j^T x_j, \quad
 \text{subject to } Ax   = b,  
\end{equation}
where $c_j^Tx_j =
\sum_{s \in K_j} \sum_{a\in \act(s)} -R(s,a) x(s,a)$,

 \[\small A = \begin{bmatrix}
  A_{00} & A_{01} &A_{02} &\ldots  & A_{0N}\\
  A_{10} & A_{11} & &  &  \\
  A_{20} & & A_{22} &  & \text{\huge0} \\
  \vdots &\text{\huge0} &  &  \ddots &  \\
  A_{N0} & &  &  & A_{NN}
\end{bmatrix}, \quad b = \begin{bmatrix}
b_0\\
b_1\\
b_2\\
\vdots\\
b_N
\end{bmatrix},\]
and $b_i \in \mathbb{R}^{n_i}$ where $b_i(k)= u_0(s)$ if $\iota_i(s)=k$. The
transformation from \eqref{eq:k0discounted} and
\eqref{eq:kidiscounted} to \eqref{eq:transformdiscountedLP} is
straightforward by rewriting the constraints and we omit the detail.


\subsection{Average-reward case}
For an ergodic \ac{mdp}, the constraints in the \ac{lp} problem of
maximizing the average reward, described by \eqref{eq:averageLP-1},
can be rewritten in the way just as how
\eqref{eq:constraintdiscounted-1} is rewritten into
\eqref{eq:k0discounted} and \eqref{eq:kidiscounted} for the
discounted-reward problem.  The difference is that for the average-reward
case, we let $\gamma=1$ and replace $u_0(s)$ with $0$, for all
$s\in S$, in \eqref{eq:k0discounted} and \eqref{eq:kidiscounted}. An
additional constraint for the average-reward case is that
$\sum_{s\in S}\sum_{a\in \act(s)} x(s,a)=1$. Hence, for an
average-reward problem in an ergodic \ac{mdp}, the corresponding LP
problem in \eqref{eq:averageLP} is formulated as
\begin{equation}
\label{eq:transformaverageLP}
\min_{x \in \mathbb{R}^{m}_+} \sum_{j=0}^N c_j^T
    x_j, \quad \text{subject to }
   \begin{bmatrix}
\mathbf{1}^T\\
A \end{bmatrix} x = \begin{bmatrix}
1\\
\bm 0
\end{bmatrix}
\end{equation}
where $\mathbf{1}^T$ is a row vector of $m$ ones,
$c_j^Tx_j = \sum_{s\in K_j} \sum_{a\in \act(s)} -R(s,a) x(s,a)$ and
the block-matrix $A$ has the same structure as of that in the
discounted case with $\gamma=1$. We can compactly write the constraint
in \eqref{eq:transformaverageLP} as $A'x=b$ where $A' $ is a sparse
constraint matrix similar in structure to the matrix $A$ in the
discounted-reward case. 
\subsection{Distributed optimization algorithm}
\label{sec:admm}
We solve the \ac{lp} problems in \eqref{eq:transformdiscountedLP} and
\eqref{eq:transformaverageLP} by employing the block splitting
algorithm based on \ac{admm} in \cite{Parikh2013}. We only present the
algorithm for the discounted-reward case in
\eqref{eq:transformdiscountedLP}. The extension to the average-reward
case is straight-forward.

First, we introduce new variables $y_i$ and
let $f_i
(y_i )= I_{\{b_i\}}(y_i)$, where for a convex set $C$,
$I_C$
is a function defined by $I_C(x)=0$
for $x\in
C$, $I_C(x)=\infty$ for $x\notin C$. Then, adding the term $f_i (y_i
)$ into the objective function enforces $y_i = b_i$.  Let $g_i(x_i) =
c_i^T x_i + I_{\mathbb{R}^{m_i}_+}(x_i)$.
The term $I_{\mathbb{R}^{m_i}_+}(x_i)$
enforces that $x_i$
is a non-negative vector. We rewrite the LP problem in
\eqref{eq:transformdiscountedLP} as follows.
\begin{align}
\label{eq:admmform}
\begin{split}
&\min_{x,y} \sum_{i=0}^{N} f_i(y_i) + \sum_{i=0}^N g_i(x_i) \\
\text{subject to }   & y_0= \sum_{i=0}^NA_{0i} x_i \text{ and for }  i =1, \ldots, N, \;\\
& y_i= A_{i0}x_0 + A_{ii}x_i.
\end{split}
\end{align}
With this formulation, we modify the block splitting algorithm in
\cite{Parikh2013} to solve \eqref{eq:admmform} in a parallel and
distributed manner (see the Appendix for
the details).
The algorithm takes parameters $\rho$,
$\epsilon^{rel}$
and $\epsilon^{abs}$:
$\rho>0$
is a penalty parameter to ensure the constraints are satisfied,
$\epsilon^{rel}>0$
is a relative tolerance and $\epsilon^{abs}>
0$ is an absolute tolerance. The choice of $\epsilon^{rel}$ and $
\epsilon^{abs}$ depends on the scale of variable values. In synthesis
of \ac{mdp}s, $\epsilon^{rel}$
and $
\epsilon^{abs}$ may be chosen in the range of $10^{-2}$ to
$10^{-6}$.
The algorithm is ensured to converge with any choice of $\rho$
and the value of $\rho$ may affect the convergence
rate. 

\section{Extension to quantitative temporal logic constraints}
\label{sec:ltl}
We now extend the distributed control synthesis methods for \ac{mdp}s
with discounted-reward and average-reward criteria to solve
Problem~\ref{prob} in which quantitative temporal logic constraints
are enforced.

\subsection{Maximizing the probability of satisfying an \ac{ltl}
  specification}
\label{subsec:maxprob}

\noindent \textbf{Preliminaries}
Given an \ac{ltl} formula $\varphi$ as the system
specification, one can always represent it by a \ac{dra}
$\calA_\varphi=\langle Q ,2^{\calAP}, T,I, \mathsf{Acc} \rangle$
where $Q$ is a finite state set, $2^{\calAP}$ is the alphabet, $I\in
Q$ is the initial state, and $T: Q\times 2^{\calAP} \rightarrow Q$
the transition function.  The acceptance condition $\mathsf{Acc}$ is a
set of tuples $\{(J_i, H_i) \in 2^Q \times 2^Q \mid i =
1,\ldots,\ell\}$. The \emph{run} for an infinite word $w= \sigma_0\sigma_1\ldots \in
(2^{\calAP})^\omega$ is the infinite sequence of states $q_0q_1\ldots
\in Q^\omega$ where $q_0=I$ and $q_{i+1}=T(q_i, \sigma_i)$ for $i\ge 0$. A run
$\rho=q_0q_1\ldots $ is accepted in $\calA_\varphi$ if there exists at
least one pair $(J_i,H_i) \in \mathsf{Acc}$ such that
$\mathsf{Inf}(\rho)\cap J_i =\emptyset$ and $\mathsf{Inf}(\rho)\cap
H_i \ne \emptyset$ where $\mathsf{Inf}(\rho)$ is the set of states
that appear infinitely often in $\rho$.

Given an \ac{mdp} $M= \langle S, A , u_0, P\rangle$ augmented with a
set $\calAP$ of atomic propositions and a labeling function
$L: S\rightarrow 2^{\calAP}$, one can compute the product \ac{mdp}
$ \mathcal{M}= M \ltimes \mathcal{A}_\varphi =\langle V,A, \Delta,v_0,
\mathsf{Acc} \rangle $
with the components defined as follows: $ V=S\times Q$ is the set of
states. $A$ is the set of actions. The initial probability
distribution of states is $\mu_0: V\rightarrow [0,1]$ such that given
$v=(s,q)$ with $q=T(I, L(s))$, it is that $\mu_0(v) = u_0(s)$.
$\Delta: V\times A \rightarrow \calD(V)$ is the transition probability
function. Given $v= (s,q)$, $\sigma$, $v'=(s',q')$ and
$q'= T(q, L(s'))$, let $\Delta (v,\sigma)(v')= P(s,\sigma)(s')$.  The
Rabin acceptance condition is
$\mathsf{Acc}=\{(\hat{J_i}, \hat{H_i}) \mid \hat{J_i} =S\times J_i,
\hat{H_i}=S\times H_i , i =1,\ldots,\ell\}$. 

By construction, a path $\rho = v_0 v_1\ldots \in V^\omega$ satisfies
the \ac{ltl} formula $\varphi$ if and only if there exists
$i\in \{1,\ldots, \ell\}$, $\Inf(\rho)\cap \hat{J}_i = \emptyset$ and
$\Inf(\rho)\cap \hat{H}_i \ne \emptyset$. To maximize the probability
of satisfying $\varphi$, the first step is to compute the set of
\emph{end components} in $\calM$, each of which is a pair $(W, f)$
where $W\subseteq V $ is non-empty and $f: W \rightarrow 2^A$ is a
function such that for any $v\in W$, for any $a\in f(v)$,
$\sum_{v'\in W}\Delta(v, a)( v')=1$ and the induced directed graph
$(W, \rightarrow_f)$ is strongly connected. Here, $v\rightarrow_f v'$
is an edge in the graph if there exists $a\in f(v)$,
$\Delta(v, a)( v') >0$.  An end component $(W,f)$ is \emph{accepting}
if $W\cap \hat{J_i} =\emptyset$ and $W\cap \hat{H_i} \ne \emptyset$
for some $i \in \{1,\ldots,\ell\}$.

Let the set of \ac{aec}s in $\calM$ be $\AEC(\calM)$ and the set of
\emph{accepting end states} be
$\calC= \{v \mid \exists (W,f)\in \AEC(\calM), v\in W\}$. Once we
enter some state $v \in \calC$, we can find an \ac{aec} $(W,f)$ such
that $v\in W$, and initiate the policy $f$ such that for some
$i\in \{1,\ldots,\ell\}$, states in $\hat{J_i}$ will be visited a
finite number of times and some state in $\hat{H_i} $ will be visited
infinitely often.

\vspace{1ex}
\noindent\textbf{Formulating the \ac{lp} problem}
An optimal policy that maximizes the probability of satisfying the
specification also maximizes the probability of hitting the set of
accepting end states $\calC$.  Reference \cite{gpu2014} develops
GPU-based parallel algorithms which significantly speed up the
computation of end components for large \ac{mdp}s. After computing the
set of \ac{aec}s, we formulate the following \ac{lp} problem to
compute the optimal policy using the proposed decomposition and
distributed synthesis method for discounted-reward cases.

  Given a product \ac{mdp}
  $\calM= \langle V, A, \Delta, \mu_0 , \acc \rangle$ and the set
  $\calC$ of accepting end states, the modified product \ac{mdp} is
  $\tilde \calM = \langle ( V \setminus \calC) \cup\{\sink\}, A,
  \tilde \Delta,\tilde\mu_0, R \rangle$
  where $(V \setminus \calC) \cup\{\sink\}$ is the set of states
  obtained by grouping states in $\calC$ as a single state
  $\sink$. For all $a\in A$, $\tilde \Delta(\sink, a)(\sink)=1$ and
  $\tilde \Delta(v, a)( \sink) = \sum_{v'\in \calC} \Delta(v,a)(v')$.
  The initial distribution $\tilde \mu_0$ of states is defined as
  follows: For $v\in V\setminus \calC$, $\tilde \mu_0(v) =\mu_0(v)$,
  and $\tilde\mu_0(\sink) = \sum_{v\in \calC} \mu_0 (v)$. The reward
  function
  $R: \left( (V \setminus \calC) \cup\{\sink\} \right)\times A
  \rightarrow \mathbb{R} $
  is defined such that for all $v$ that is not $\sink$,
  $R(v, a) = \sum_{v' \in (V \setminus \calC) \cup\{\sink\}} \tilde
  \Delta(v,a)( v') \cdot \bm 1_{\{\sink\}}(v')$
  where $\bm 1_{X}(v)$ is the indicator function that outputs $1$ if
  and only if $v\in X$ and $0$ otherwise. For any action
  $a\in A(\sink)$, $R(\sink,a ) =0$.

  By definition of reward function, the discounted reward
  with $\gamma=1$ from state $v$ in the modified product \ac{mdp}
  $\tilde \calM$ is the probability of reaching a state in $\calC$
  from $v$ under policy $f$ in the product \ac{mdp} $\calM$.  Hence,
  with a decomposition of $\calM$, the proposed distributed synthesis
  method for discounted-reward problems can be used to compute the
  policy that maximizes the probability of satisfying a given \ac{ltl}
  specification.

\subsection{Average reward under B\"uchi acceptance conditions}
\noindent \textbf{Preliminaries} Consider a temporal logic formula $\varphi$ that can be expressed as a
\ac{dba} $\calA_\varphi=\langle Q ,2^{\calAP}, T,I, F_\varphi \rangle$
where $Q, 2^{\calAP}, T, I$ are defined similar to a \ac{dra} and
$F_\varphi \subseteq Q$ is a set of \emph{accepting states}. A run
$\rho$ is accepted in $\calA_\varphi$ if and only if
$\Inf(\rho)\cap F_\varphi \ne \emptyset$.  Given an \ac{mdp}
$M =\langle S, A, u_0, P, \calAP, L \rangle$ and a \ac{dba}
$\mathcal{A}_\varphi= \langle Q ,2^{\calAP}, T, q_0, F_\varphi
\rangle$,
the \emph{product \ac{mdp} with B\"uchi objective} is
$ \mathcal{M}= M \ltimes \mathcal{A}_\varphi =\langle V, A, \Delta,
\mu_0, F \rangle $
where components $V, A, \Delta, \mu_0$ are obtained similarly as in
the product \ac{mdp} with Rabin objective.  The difference is that
$F \subseteq S\times F_\varphi$ is the set of accepting states. A path
$\rho = v_0 v_1\ldots \in V^\omega$ satisfies the \ac{ltl} formula
$\varphi$ if and only if $\Inf(\rho)\cap F \ne \emptyset$. 

\vspace{1ex}
\noindent\textbf{Formulating the \ac{lp} problem} For a product \ac{mdp}
$\calM$ with B\"uchi objective, we aim to synthesize a policy that
maximizes the expected frequency of visiting an accepting state in the
product \ac{mdp} $\mathcal{M} = M\ltimes \calA_\varphi$.  This type of
objectives ensures some recurrent properties in the temporal logic
formula are satisfied as frequently as possible. For example, one such
objective can be requiring a mobile robot to maximize the frequency of
visiting some critical regions.

This type of objectives can be formulated as an average-reward problem
in the following way: Let the reward function
$R: V\times A\rightarrow \mathbb{R}$ be defined by
$R(v,a ) = \sum_{v'\in V} \Delta(v,a)(v') \cdot \mathbf{1}_{ F }(v')$.
By definition of the reward function, the optimal policy with respect
to the average-reward criterion is the one that maximizes the
frequency of visiting a state in $F$. If the product \ac{mdp} is
ergodic, we can then solve the resulting average-reward problem by the
distributed optimization algorithm with a decomposition of product
\ac{mdp} $\calM$.


\section{Case studies}
\label{sec:examples}
We demonstrate the method with three robot motion planning examples. 
All the experiments were run on a machine with Intel Xeon 4 GHz,
8-core CPU and 64 GB RAM running Linux. The distributed optimization
algorithm is implemented in MATLAB. The decomposition and other
operations are implemented in Python.

\begin{figure}[ht]
\centering
\begin{subfigure}[b]{0.15\textwidth}
\includegraphics[width=\textwidth]{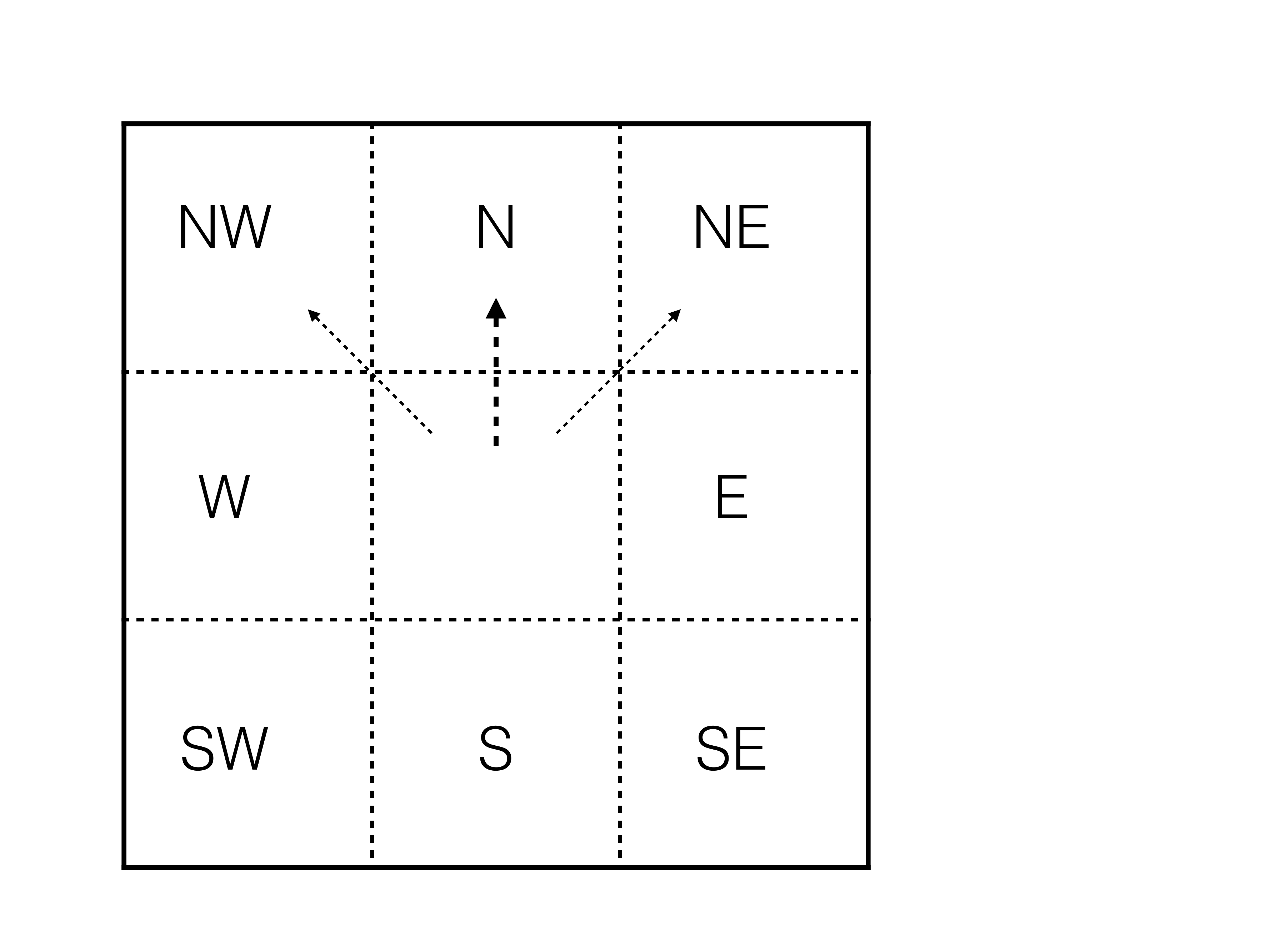}
\caption{}
\label{fig:singlestep}
\end{subfigure}
\begin{subfigure}[b]{0.32\textwidth}
\includegraphics[trim = 0mm 25mm 0mm 20mm, clip, width=\textwidth]{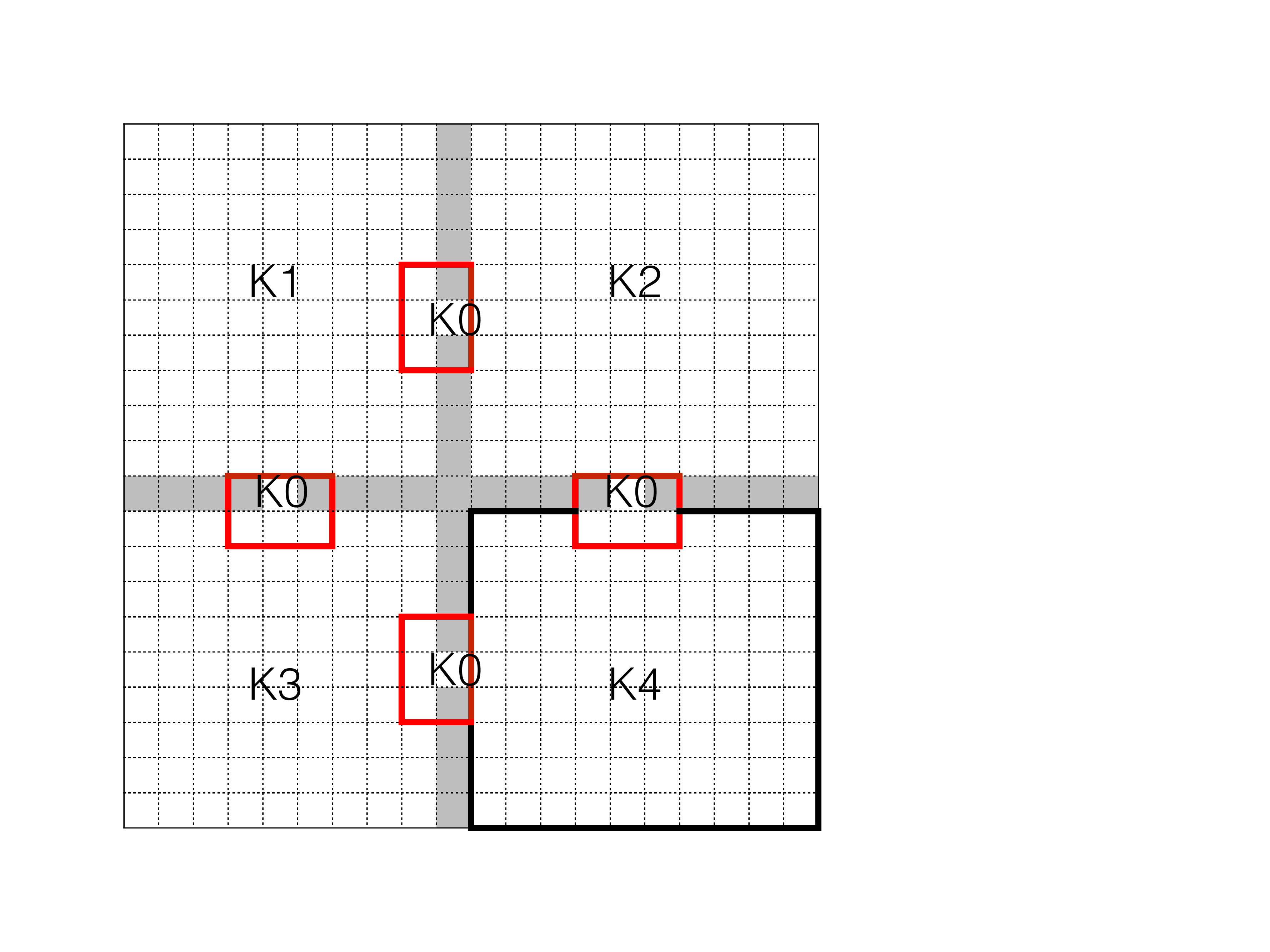}
\caption{}
\label{fig:decompgw}
\end{subfigure}
\label{fig:gridworld2}
\caption{(a) A fraction of a $m\times n $ gridworld. The dash arrow
  represents that if the robot takes action `N', there are non-zero
  probabilities for it to arrive at NW, N, and NE cells. (b) A
  $20\times 20$ gridworld. A natural partition of state space using
  the walls gives rise to $K_0, K_1,K_2,K_3,K_4$ subsets of
  states. States in $K_0$ are enclosed using the squares. }
\end{figure}

Figure~\ref{fig:singlestep} shows a fraction of a gridworld. A robot
moves in this gridworld with uncertainty in different terrains
(`grass', `sand', `gravel' and `pavement'). In each terrain and for
robot's different action (heading north (`N'), south (`S'), west (`W')
and east (`E')), the probability of arriving at the correct cell is
$0.9$ for pavement, $0.85$ for grass, $0.8$ for gravel and $0.75$ for
sand. With a relatively small probability, the robot will arrive at
the cell adjacent to the intended one.
Figure~\ref{fig:decompgw} displays a $20\times 20 $ gridworld. The
grey area and the boundary are walls. If the robot runs into the wall,
it will be bounce back to its original cell. The walls give rise to a
natural partition of the state space, as demonstrated in this
figure. 
If no
explicit modular structure in the system can be found, one can compute
a decomposition using the method in section~\ref{sec:decompmethod}.
In the following example, the wall pattern is the same as in the
$20\times 20 $ gridworld.

\subsection{Discounted-reward case} 

We select a subset $W$ of cells as ``restricted area'' and a subset
$G$ of cells as ``targets''. The reward function is given: For
$s\in S$, $S\notin G \cup W $, $R(s,a)= -1$ counts for the amount of
time the robot takes action $a$.  For $s\in W$, for all $a\in A(s)$,
$R(s,a)=-1000$. For $s\in G$, $R(s,a )=100$ for all $a\in A(s)$.
Intuitively, this reward function will encourage the robot to reach
the target with as fewer expected number of steps as possible, while
avoiding running into a cell in the restricted area. We select
$\gamma= 0.9$.

\textbf{Case 1: } To show the convergence and correctness of the distributed
optimization algorithm, we first consider a $100\times 100$ gridworld
example that can be solved directly with a centralized
algorithm. Since at each cell there are four actions for the robot to
pick, the total number of variables is $4\times 8220$ for the
$100\times 100$ gridworld (the wall cells are excluded from the set of
states).  In this gridworld, there is only $1$ target cell. The
restricted area include $50$ cells.  The resulting \ac{lp} problem
\eqref{eq:constraintdiscounted} can be solved using CVX, a package for
specifying and solving convex programs \cite{cvx}. The problem is
solved in $4.77$ seconds, and the optimal objective value under the
optimal policy given by CVX is $10$. 

Next, we solve the same problem by decomposing the state space
of the \ac{mdp} along the walls into $25$ regions, each of which is a
$20\times 20$ gridworld. This partition of state space yields $75$
states for each $K_i, i>0$ and $720 $ states for $K_0$. In which
follows, we select $\rho= 80, 100,200, 500, 1000$ to show the
convergence of the distributed optimization algorithm. Irrespective of
the choices for $\rho$, the average time for each iteration is about
$0.16$ sec.  The solution accuracy relative to CVX is summarized in
Table~\ref{tbl:summary}.  The `rel. error in objval' is the relative
error in objective value attained, treating the CVX solution as the
accurate one, and the infeasibility is the relative primal
infeasibility of the solution, measured by
$ \frac{\norm{Ax^\ast -b}_2 }{1+ \norm{b}_1}.$
Figure~\ref{fig:reobjval100} shows the convergence of the optimization
algorithm.

\begin{table*}[t]
\caption{Relative resolution accuracy for solving the $100\times 100$
  gridworld example with discounted reward (under $\epsilon^{abs}= 10^{-5}$, $\epsilon^{rel}= 10^{-4}$).}
\centering
\begin{tabular}{c| c c c c c c}
$\rho$  & 80 & 100 & 200  & 500 & 1000 \\
Iterations   & 12001 &11014 & 13868 & 11866 & 12733\\
objval  & 9.54 & 9.90&  9.89 & 9.96 & 9.96 \\
rel. error ($\%$)&4.6& 1.0& 1.1 & 0.38 & 0.38 \\ 
infeasibility   &  $2.7\times 10^{-3}$ & $2\times 10^{-3}$   & $0.885  \times 10^{-3}$
               & $0.45 \times 10^{-3}$  & $0.23 \times 10^{-3}$
\end{tabular}
\label{tbl:summary} \end{table*}
\begin{figure}[t]
\centering
\includegraphics[width=0.45\textwidth]{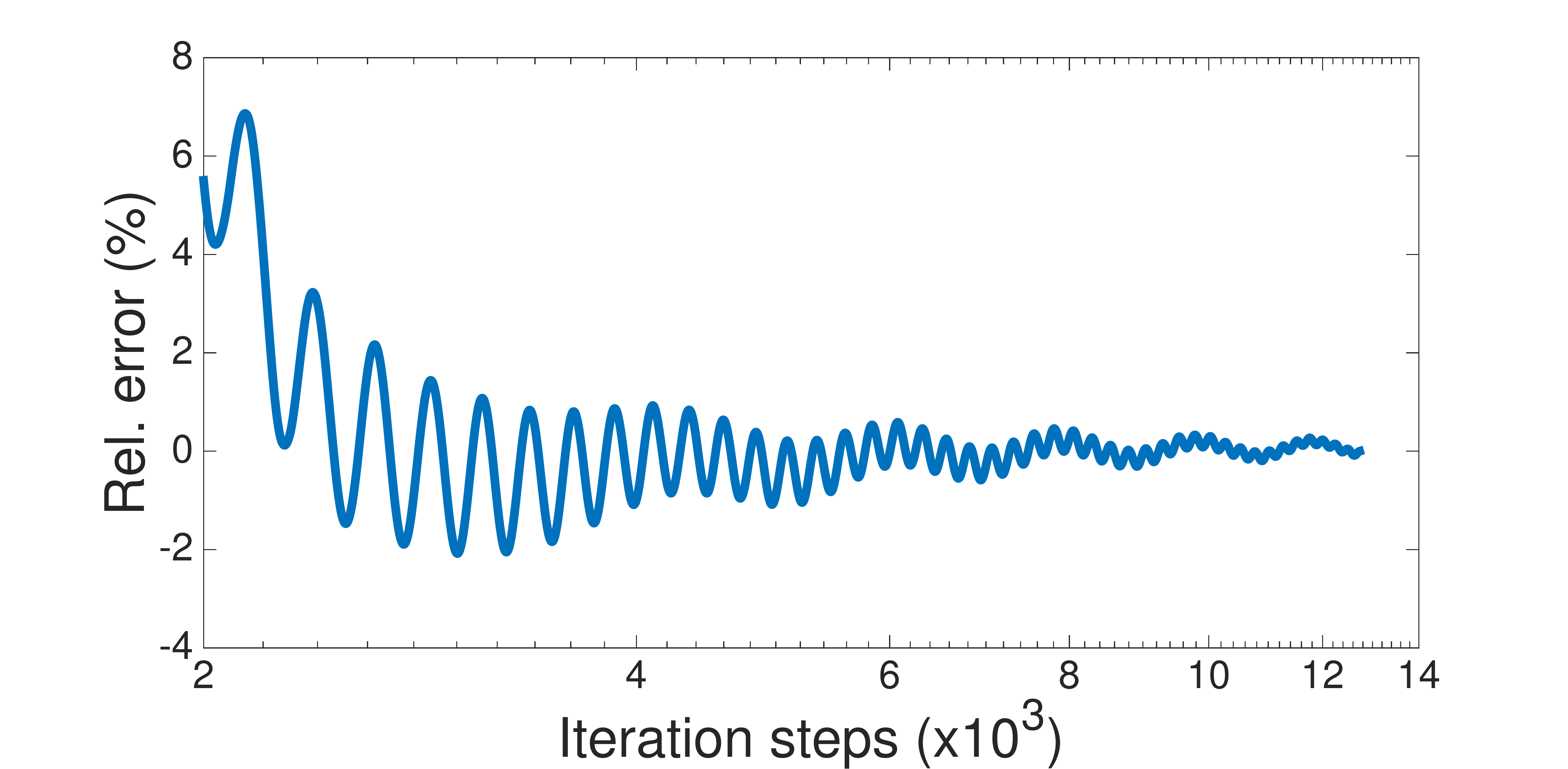}
\caption{Relative error in the objective value vesus iterations in
  $100\times 100$ gridworld with discounted reward, under
  $\rho=1000$. For clarity, we did not draw the relative error for the
  initial $2000$ steps, which are comparatively large.}
\label{fig:reobjval100}
\end{figure}

\textbf{Case 2:} For a $1000\times 100$ gridworld, the centralized
method in CVX fails to produce a solution for this large-scale
problem. Thus, we consider to solve it using the decomposition and
distributed synthesis method. In this example, we partition the
gridworld such that each region has $50\times 50$ cells, which results
in $40$ regions. There are $1160$ states in $K_0$ and about $2005$
states in each $K_i$, for $i = 1,\ldots, 40$. In this example, we
randomly select $40$ cells to be the targets and $40$ cells to be the
restricted areas.  By choosing $\rho=1000$, $\epsilon^{rel}=10^{-6}$,
$\epsilon^{abs} =10^{-6}$, the optimal policy is solved within $25342$
seconds and it takes about $2.6$ seconds for one iteration. The total
number of iterations is $9747$. Under the (approximately)-optimal
policy obtained by distributed optimization, the objective
value is $138.73$.  The relative primal infeasibility of the solution
is $0.29\times10^{-4}$.  Figure~\ref{fig:objval1000} shows the
convergence of distributed optimization algorithm.

\begin{figure}[t]
\begin{subfigure}[b]{0.45\textwidth}
\centering
\includegraphics[width=1\textwidth]{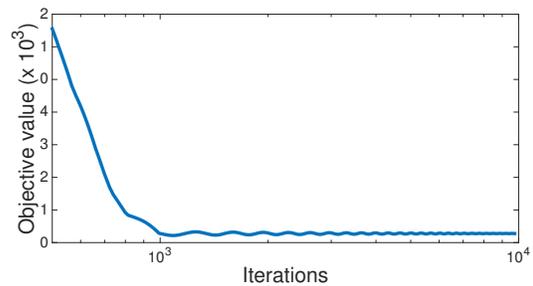}
\caption{}
\label{fig:objval1000}
\end{subfigure}
\begin{subfigure}[b]{0.45\textwidth}
\centering
\includegraphics[width=1\textwidth]{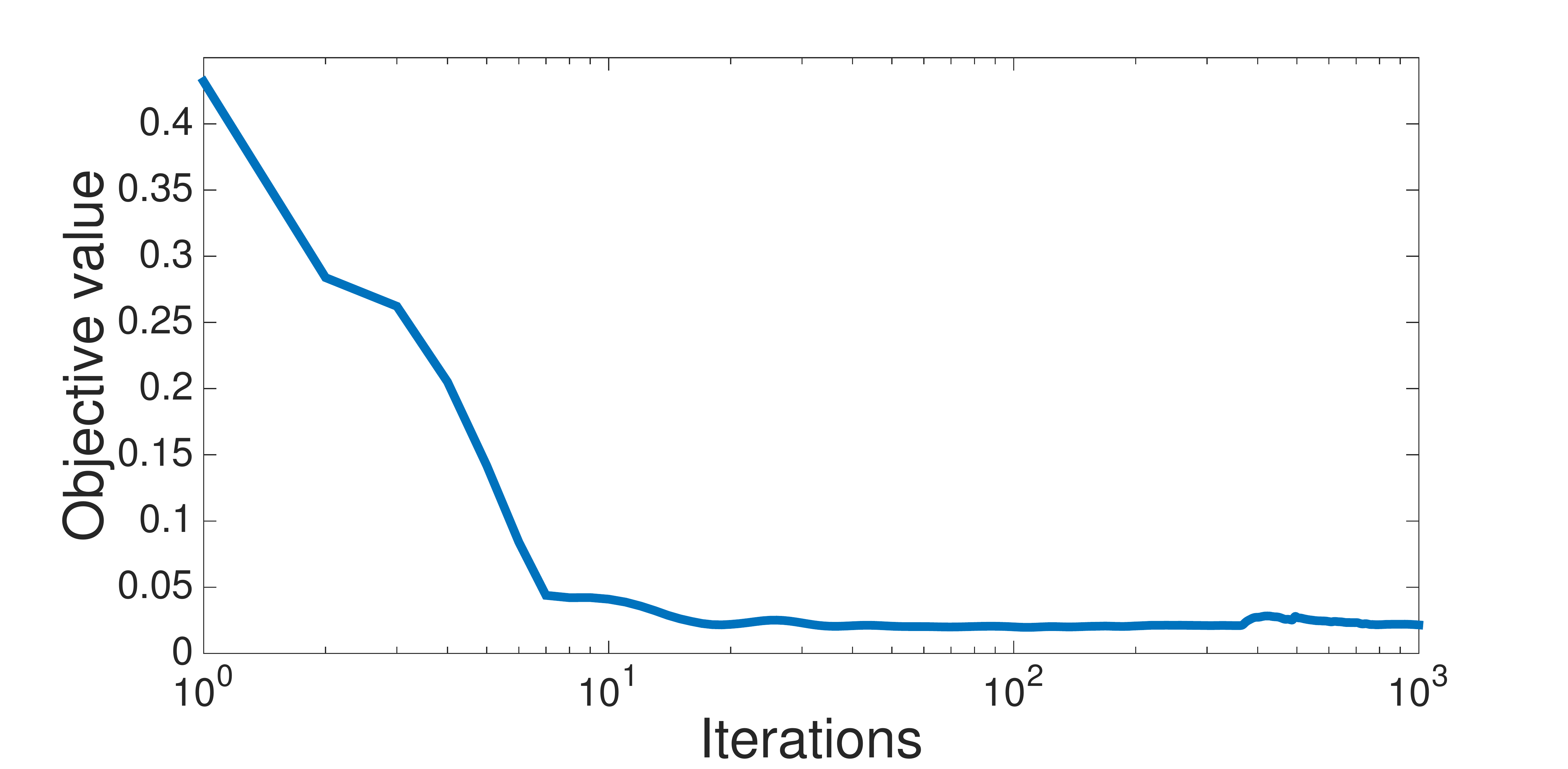}
\caption{}
\label{fig:reobjval50}
\end{subfigure}
\caption{ Objective value vesus iterations in (a) $1000 \times 100$
  gridworld (the initial $500$ steps are omitted). (b) Objective value
  vesus iterations in $50 \times 50$ gridworld with a B\"uchi
  objective. Here we only show the first $1000$ iterations as the
  objective value converges to the optimal one after $1000$ steps. }
\label{fig:relobjval1000a50}
\end{figure}

\subsection{Average-reward case with quantative temporal logic
  objectives}
We consider a $50 \times 50$ gridworld with no obstacles and $4$
critical regions labeled ``$R_1$'' , ``$R_2$'', ``$R_3$'' and
``$R_4$''. The system is given a temporal logic specification
$\varphi\coloneqq \square \lozenge (R_1\land \lozenge R_2) \land
\square \lozenge (R_3\land \lozenge R_4) $.
That is, the robot has to always eventually visit region $R_1$ and
then $R_2$, and also always eventually visit region $R_3$ and then
$R_4$. The number of states in the corresponding \ac{dba} is $14$
after trimming the unreachable states, due to the fact that the robot
cannot be at two cells simultaneously. The quantitative objective is
to maximize the frequency of visiting all four regions (an accepting
state in the \ac{dba}). The formulated \ac{mdp} is ergodic and
therefore our method for average-reward problems applies.

For an average-reward case, we need to satisfy the
constraint $\mathbf{1}^T x=1$ in \eqref{eq:transformaverageLP}. This
constraint leads to slow convergence and policies with large
infeasibility measures in distributed optimization. To handle this
issue, we approximate average reward with discounted reward
\cite{kakade2001optimizing}: For ergodic \ac{mdp}s,
the discounted accumulated reward, scaled by $1-\gamma$, is
approximately the average reward. Further, if $\frac{1}{1-\gamma}$ is
large compared to the mixing time \cite{puterman2009markov} of the
Markov chain, then the policy that optimizes the discounted
accumulated reward with the discounting factor $\gamma$ can achieve an
approximately optimal average reward. 


Given $\rho=1000$, $\epsilon^{ref}=10^{-5}$ and
$\epsilon^{abs}=10^{-5}$, $\gamma=0.98$, the distributed synthesis
algorithm terminates in $14284$ iteration steps and the optimal
discounted reward is $0.9998$.  Scaling by $1-\gamma=0.02$, we obtain
the average reward $0.9998\times 0.02= 0.02$, which is the approximately
optimal value for this average reward under the obtained policy.  The
convergence result is shown in Figure~\ref{fig:reobjval50} and the
infeasibility measure of the obtained solution is $0.016$.

\section{Conclusion}
\label{sec:conclude}
For solving large Markov decision process models of stochastic
systems with temporal logic specifications, we developed a
decomposition algorithm and a distributed synthesis
method. This decomposition exploits the modularity in the system
structure and deals with sub-problems of smaller sizes. We employed
the block splitting algorithm in distributed optimization based on the
alternating direction method of multipliers to cope with the
difficulty of combining the solutions of sub-problems into a solution
to the original problem. Moreover, the formal decomposition-based
distributed control synthesis framework established in this paper
facilitates the application of other distributed and parallel
large-scale optimization algorithms \cite{Scutari2014} to further
improve the rate of convergence and the feasibility of solutions for
control synthesis in large \ac{mdp}s. In the future, we will develop
an interface to PRISM toolbox \cite{KNP11} with an implementation of
the proposed decomposition and distributed synthesis
algorithms.
\appendix

\section{Appendix} 
\label{App}

At the $k$-th iteration, for $ i,j=0,\ldots, N$,
\begin{align*}
\label{eq:admm}
\begin{split}
y_{i}^{k+1/2} & : = \prox_{f_i}(y_i^k  - \tilde y_i^k) = b_i, \\
x_j^{k+1/2} & : = \prox_{g_j}(x_j^k  - \tilde x_j^k)  \\
& =
\proj_{\mathbb{R}_+^{m_j}} (x_j^k - \tilde x_j^k - c_j/\rho), \\
(x_{ij}^{k+1/2}, y_{ij}^{k+1/2} ) & := \proj_{ij} (x_j^k - \tilde
x_{ij}^k, y_{ij}^k +\tilde y _{i}^k), \\
x_j^{k+1}& := \avg(x_j^{k+1/2}, \{x_{ij}^{k+1/2}
\}_{i=0}^{N}), \\
(y_i^{k+1}, \{y_{ij}^{k+1}\}_{j=0}^{N}) & := \exch(y_i^{k+1/2},
\{ y_{ij}^{k+1/2}\}_{j=0}^{N}),\\
& \text{ if } i=0, \\
(y_i^{k+1}, \{y_{i0}^{k+1}, y_{ii}^{k+1}\}) & := \exch(y_i^{k+1/2},
\{ y_{i0}^{k+1/2}, y_{ii}^{k+1/2}\}),\\
& \text{ if } i=1,\ldots,N,\\
\tilde x_j^{k+1} & := \tilde x_j^k+ x_j^{k+1/2} -x_j^{k+1},\\
\tilde y_i^{k+1} & := \tilde y_i^k+ y_i^{k+1/2} -y_i^{k+1},\\
\tilde x_{ij}^{k+1} & := \tilde x_{ij}^k+ x_{ij}^{k+1/2} -x_j^{k+1},
\end{split}
\end{align*}
where $\proj_{R^{m_i}_+}$ denotes the projection to the nonnegative
orthant, $\proj_{ij}$ denotes projection onto
$\{(x,y) \mid y=A_{ij}x\}$.  $\avg$ is the elementwise averaging
\footnote{Since for some $i,j$, $x_{ij}^{k+1/2} =0$, in the
  elementwise averaging, these $x_{ij}^{k+1/2}$ will not be
  included.}; and $\exch$ is the exchange operator, defined as below.
$ \exch(c, \{c_j\}_{j=1}^N) $ is given by
$y_{ij}:= c_j + (c -\sum_{j=1}^N c_j)/(N-1)$ and
$y_i : = c - (c-\sum_{j=1}^Nc_j)/N-1$. The variables can be
initialized to $0$ at $k=0$. Note that the computation in each
iteration can be parallelized.  The iteration terminates when the
stopping criterion for the block splitting algorithm is met (See
\cite{Parikh2013} for more details).  The solution can be obtained
$x^\ast = (x_0^{k+1/2},\ldots, x_N^{k+1/2})$.
\bibliographystyle{ieeetran}

\end{document}